\pdfoutput=1
\documentclass[letterpaper]{article}

\usepackage{hyperref}
\usepackage{amsmath}
\usepackage{amsthm}
\usepackage{amsfonts}
\usepackage{url}
\usepackage{gnuplot-lua-tikz}

\newtheorem{theorem}{Theorem}
\newtheorem{lemma}{Lemma}
\newtheorem{corollary}{Corollary}

\def\edgevector{{\mathbf X}}
\def\numvars{k}

\begin{document}

\title{A Bound on the Expected Optimality\\of Random Feasible Solutions\\to Combinatorial Optimization Problems}
\author{Evan A. Sultanik\\The Johns Hopkins University APL\thanks{Evan is no longer affiliated with Johns Hopkins University.  He is currently the chief scientist at Digital Operatives and adjunct faculty at the Drexel University Department of Computer Science.}\\\url{evan@sultanik.com}\\\url{http://www.sultanik.com/}}

\maketitle

\begin{abstract}
This paper investigates and bounds the expected solution quality of
combinatorial optimization problems when feasible solutions are chosen
at random.  Loose general bounds are discovered, as well as families
of combinatorial optimization problems for which random feasible
solutions are expected to be a constant factor of optimal.  One
implication of this result is that, for graphical problems, if the
average edge weight in a feasible solution is sufficiently small, then
any randomly chosen feasible solution to the problem will be a
constant factor of optimal.  For example, under certain well-defined
circumstances, the expected constant of approximation of a randomly
chosen feasible solution to the Steiner network problem is bounded
above by 3.  Empirical analysis supports these bounds and actually
suggest that they might be tightened.
\end{abstract}

\section{Introduction}

Combinatorial optimization is the process of selecting a
set of points from a finite topological space that maximize or
minimize a given objective function.
It is often very simple to find \emph{feasible} solutions to such
problems, however, it is usually exceedingly hard to
find \emph{optimal} ones.  For example, given a finite set of objects
$y \in Y$ each with value $v : Y \rightarrow \mathbb{R}$ and weight $w
: Y \rightarrow \mathbb{R}_{\geq 0}$, the knapsack problem asks to
find a subset of the objects $S \subseteq Y$ whose combined weight
does not exceed a given maximum, $w_{\text{max}}$, and whose combined
value is maximized:
\begin{displaymath}
  \begin{aligned}
     \text{maximize\ } & \sum_{y \in Y} v(y) x_y  & \\
     \text{subject to:\ }\\     &  \sum_{y \in Y} w(y) x_y \leq w_{\text{max}}, &\\
                    &    x_y \in \{0, 1\}, & \forall y \in Y,
  \end{aligned}
\end{displaymath}
where the chosen set of objects is $S = \{y \in Y : x_y = 1\}$.
Finding a feasible solution to the knapsack problem is trivial: Any
set of objects $S \subseteq Y$ that does not outweigh the maximum,
\begin{displaymath}
  \sum_{s \in S} v(s) \leq w_{\text{max}},
\end{displaymath}
will be a feasible solution.  Such a set of objects can be chosen in
linear time.

As a second example, consider another $\mathbf{NP}$-\textsc{Complete}
problem: the Steiner network problem~\cite{goemans94improved}.  Given
an undirected, edge-weighted graph on $n$ vertices and $k$ edges, $G
= \langle V, E\rangle$, and also given a subset of ``terminal''
vertices $T \subseteq V$ of cardinality $\alpha = |T|$, the Steiner
network problem asks to find a minimum weight forest that spans the
terminals~\cite{aggarwal94scaling}:
\begin{displaymath}
  \begin{aligned}
     \text{minimize\ } & \sum_{e \in E} w(e) x_e & \\
     \text{subject to:\ }\\     &    \sum_{e \in \delta(S)} x_e \geq f(S), & \forall S \subset V : S \neq \emptyset\\
                    &    x_e \geq 0, & \forall e \in E,
  \end{aligned}
\end{displaymath}
where $\delta(S)$ is the set of edges having exactly one endpoint in
$S$ and $f : 2^V \rightarrow \{0, 1\}$ is a function such that $f(S) =
1$ if and only if $\emptyset \neq S \cap T \neq T$.  Despite the fact
that finding an optimal solution is hard, finding a \emph{feasible}
solution is not: Any spanning forest will be a feasible solution to
the Steiner network problem.

Given that finding a feasible solution is often so simple relative to
finding the optimal, we would like to know: In expectation, how much
worse is a given feasible certificate compared to the optimal
solution?  This paper investigates and bounds the expected solution
quality to such problems when feasible solutions are chosen at random.

We proceed by bounding the expected value of the best- and worst-case
outcomes of any choice of objects.  In general, let $m$ be a lower
bound on the number of objects in an optimal solution.  In the case of
the Steiner network problem, for example, $m$ must be greater than or
equal to $\left\lfloor\frac{\alpha}{2}\right\rfloor$ edges.  Let
$\ell$ be an upper bound on the number of objects chosen by a given
algorithm that is guaranteed to produce feasible---but not necessarily
optimal---solutions.  In the worst case, the algorithm will have
chosen the $\ell$ costliest objects while the optimal solution
actually contained the $m$ cheapest objects.  How likely is this
outcome?  If we can bound the ratio of the expected value of the sum
of the $\ell$ costliest objects over the expected value of the sum of
the $m$ cheapest objects, this will be an upper bound on the expected
constant of approximation of the algorithm, \emph{regardless} of how
the algorithm actually chooses a feasible solution, and also
regardless of which combinatorial optimization problem we are solving.

The bound on the expected constant of approximation is codified in the
remainder of this paper.  Section~\ref{sec:framing} serves to
formalize the problem and Section~\ref{sec:bounds} develops the bound,
the theoretical consequences of which are discussed in
Section~\ref{sec:discussion}.  The theoretical bound is supported
through empirical evaluation in Section~\ref{sec:evaluation}.  Proofs
of the various lemmas and theorems are provided at the end of the
paper in~\ref{appendix:proofs}.

\section{Framing the Problem}
\label{sec:framing}

It is often useful to think of combinatorial optimization as the
process of selecting an optimal subset from a given finite set of
objects.  We proceed by treating the cost of each object in a
combinatorial optimization problem as a random variable.  Let
$\edgevector = [X_1, \ldots, X_\numvars]$ be a vector random variable
consisting of identically distributed (but not necessarily
independent) random variables.  In a sense, each variable captures
whether or not one of the $k$ objects from the combinatorial
optimization problem will be chosen to be a part of the solution.  In
the knapsack problem, for example, $k$ is the total number of objects.
In the Steiner network problem, $k$ is the number of edges in the
graph.  The order statistics of $\edgevector$ are denoted
$X_{(1:\numvars)} \leq \ldots \leq X_{(r:\numvars)} \leq \ldots \leq
X_{(\numvars:\numvars)}$.  For $\ell \in \{1, 2, \ldots, \numvars\}$,
let $Y$ be the distribution of the $\ell$ largest order statistics of
$\edgevector$:
\begin{displaymath}
  Y \sim \sum_{r=\numvars-\ell+1}^\numvars X_{(r:\numvars)}.
\end{displaymath}
Similarly, for $m \in \{1, 2, \ldots, \numvars\}$, let $Y^*$ be the
sum of the $m$ smallest order statistics of $\edgevector$:
\begin{displaymath}
 Y^* \sim \sum_{r=1}^m X_{(r:\numvars)}.
\end{displaymath}
Without loss of generality, we assume in the remainder of our analysis
that the objective function is being minimized.  In that case, observe
that $Y$ is the probability distribution of the cost of the costliest
solution of size $\ell$ and, similarly, $Y^*$ is the distribution of
the cost of the best possible solution of size $m$.  The expected
value of $Y$ will be an upper bound on the expected value of the
costliest solution to the problem, whereas the expected value of $Y^*$
will be a lower bound on the expected value of the optimal solution.

We are interested in answering the question: What can be said of the
relationship between $Y$ and $Y^*$?  If $Y$ and $Y^*$ are
nonnegative\footnote{\textit{i.e.}, their distributions are truncated
such that the probability density for all negative values is zero.}
then the expected value of $Y/Y^*$ will be an upper bound on the
expected constant of approximation of any algorithm that produces
feasible solutions to the optimization problem.  In general, for arbitrary $Y$ and $Y^*$, the
expected approximation factor is bounded above by
\begin{equation}
  \label{equ:approxbound}
  1 + \left|\frac{E[Y] - E[Y^*]}{E[Y^*]}\right|.
\end{equation}
We develop an upper bound on this expression in the following section.

\section{Bounds on Order Statistics}
\label{sec:bounds}

Gascuel and Caraux discovered the following bounds on the expected value of order statistics of random variables with symmetrical distributions with mean $\mu$ and variance $\sigma^2$:
\begin{theorem}[Proposition 2 of~\cite{gascuel92bounds}]
  \label{thm:origbounds}
  The bounds
  \begin{displaymath}
    \left.\begin{array}{ll}
      \text{if}\ \frac{r}{\numvars}\leq \frac{1}{2}, & \mu - \sigma\sqrt{\frac{\numvars}{2r}}\\
      \text{if}\ \frac{r}{\numvars}\geq \frac{1}{2}, & \mu - \sigma\sqrt{\frac{\numvars(\numvars-r)}{2r^2}}
    \end{array}\right\}
    \leq E[X_{(r:\numvars)}] \leq
    \left\{\begin{array}{ll}
      \mu + \sigma\sqrt{\frac{\numvars}{2(\numvars-r+1)}} & \text{if}\ \frac{r-1}{\numvars} \geq \frac{1}{2},\\
      \mu + \sigma\sqrt{\frac{\numvars(r-1)}{2(\numvars-r+1)^2}} & \text{if}\ \frac{r-1}{\numvars} \leq \frac{1}{2},
    \end{array}\right.
  \end{displaymath}
  are valid and may be reached for some distributions.
\end{theorem}
With the result of Theorem~\ref{thm:origbounds} and some intermediary
bounds on the generalized harmonic series, we are led directly to a
loose lower bound for $E[Y^*]$ and a loose upper bound for $E[Y]$.

To achieve these bounds, let us first constrain the magnitude of $m$
and $\ell$ with respect to $\numvars$:
\begin{displaymath}
  \frac{m}{\numvars} \leq \frac{1}{2} \wedge \frac{\ell-1}{\numvars} \leq \frac{1}{2}.
\end{displaymath}
These constraints appear to be reasonable: The family of combinatorial
optimization problems that conform to these constraints is large.  In
the context of the Steiner network problem, for example, these
constraints equate to requiring that both the optimal and feasible
solutions encompass fewer than half of all possible edges in the
graph.  This assumption clearly holds for almost every simple graph
since
\begin{enumerate}
  \item such graphs have ${n \choose 2} = \frac{n}{2}(n-1)$ possible edges;
  \item any feasible solution to the Steiner network problem will have at most $n-1$ edges; and
  \item $\forall n \geq 4 : n-1 \leq \frac{n}{4}(n-1)$.
\end{enumerate}

The constraints on the magnitude of $m$ and $\ell$ lead to the following loose bounds:
\begin{theorem}[Proof of this and the remainder of the theorems in this paper are given in~\ref{appendix:proofs}]
  \label{thm:mlb}
  \begin{displaymath}
    \frac{m}{\numvars} \leq \frac{1}{2} \implies E[Y^*] > m\mu - \sigma\sqrt{2\numvars}\left(\sqrt{m+1} - 1\right).
  \end{displaymath}
\end{theorem}
\begin{theorem}
  \label{thm:lub}
  \begin{displaymath}
    \frac{\ell - 1}{\numvars} \leq \frac{1}{2} \implies E[Y] \leq\ell\mu + \frac{\sigma\sqrt{2\numvars}}{2}\left(2\sqrt{\ell} - 1\right).
  \end{displaymath}
\end{theorem}
It is possible to bound the expected approximation factor of Equation~(\ref{equ:approxbound}) by utilizing the bounds on $E[Y^*]$ and $E[Y]$:
\begin{theorem}
  \label{thm:finalbounds}
  If $2(\ell - 1) \leq k$ and $2m \leq k$ then
  \begin{displaymath}
    1 + \left|\frac{E[Y] - E[Y^*]}{E[Y^*]}\right| \leq
    \begin{cases}
      \displaystyle\frac{\ell\mu - 2m\mu + 3\sigma\numvars}{\sigma\numvars - m\mu}, & \text{if}\ \mu < \sigma\sqrt{2k}\frac{\sqrt{m+1}-1}{m}\\
      &\\
      \displaystyle
      \frac{\ell\mu + \sigma k}{m\mu - \sigma k}, & \text{if}\ \mu > \sigma\sqrt{2k}\frac{\sqrt{m+1}-1}{m},
    \end{cases}
  \end{displaymath}
  which simplifies to
  \begin{displaymath}
    \leq\begin{cases}
      2, & \text{if}\ \mu < \sigma\sqrt{2k}\frac{\sqrt{m+1}-1}{m} \wedge \left(\ell\mu \leq -\sigma k \vee m\mu > \sigma k\right)\\
      3, & \text{if}\ \mu \leq 0\\
      4, & \text{if}\ \mu(\ell + 2m) \leq \sigma k\\
      \varepsilon \in \mathbb{R}_{\geq 2}, & \text{if}\ \mu \leq \frac{\sigma k(\varepsilon - 3)}{\ell - 2m + \varepsilon m} < \sigma\sqrt{2k}\frac{\sqrt{m+1}-1}{m}\\
      \varepsilon \in \mathbb{R}_{\geq k}, & \text{if}\ \sigma\sqrt{2k}\frac{\sqrt{m+1}-1}{m} < \mu \leq \frac{\sigma(\varepsilon-k)}{\ell}.
    \end{cases}
  \end{displaymath}
\end{theorem}
\begin{corollary}
  \label{cor:ub}
  If
  \begin{displaymath}
    \ell > \exp\left(2\sqrt{m+1}-3\right) \wedge \mu > \sigma\sqrt{2\numvars}\frac{\sqrt{m+1}-1}{m} \wedge m \leq \ell \leq \frac{\numvars}{2},
  \end{displaymath}
  then the upper bound simplifies to $3$.
\end{corollary}
\begin{corollary}
\label{cor:asymp}
If 
$\sigma \gg \mu$, the constant bound of 3 will hold asymptotically.
\end{corollary}

\section{Discussion}
\label{sec:discussion}

Theorem~\ref{thm:finalbounds} is a rather surprising result: If we
know that the size of the optimal solution is bounded below by $m$ and the
average edge weight in a random feasible solution is small,
then \emph{any} randomly chosen solution of size at most $\ell$ will,
on average, be a constant factor of optimal.  If the edges are
weighted from a distribution with a nonpositive mean, then any
randomly chosen feasible solution of any size is guaranteed to be, on
average, no more than three times the cost of the optimal solution.

Consider the Steiner network problem as an example.  If there are
$\alpha$ terminals then we know that the optimal solution must have at
least $\left\lfloor\frac{\alpha}{2}\right\rfloor$ edges (this is $m$).
Any feasible solution to the problem is going to be an acyclic graph,
which will have at most $n-1$ edges (this is $\ell$).  The number of
random variables, $\numvars$, is equal to the number of edges, which
is bounded above by ${n \choose 2}$.  If $\mu$ is
small\footnote{\textit{i.e.}, the first case of the piecewise bound of
Theorem~\ref{thm:finalbounds} is satisfied.} then, by
Theorem~\ref{thm:finalbounds}, the expected constant of approximation
for any randomly chosen feasible solution to the Steiner network
problem is bounded above by
\begin{displaymath}
  \frac{\mu(n-1) + 3\sigma{n \choose 2} - 2\mu\left\lfloor\frac{\alpha}{2}\right\rfloor}{\sigma{n \choose 2} - \mu\left\lfloor\frac{\alpha}{2}\right\rfloor},
\end{displaymath}
which quickly\footnote{Convergence is superlinear, as confirmed by d'Alembert's ratio test.} converges to $3$ as $n \rightarrow \infty$.


\section{Empirical Evidence}
\label{sec:evaluation}

In order to demonstrate the validity and tightness of the bound, we
empirically evaluate the optimality of randomly chosen feasible
solutions to combinatorial optimization problems.  First, let us
consider a relatively simple problem that is in $\mathbf{P}$: minimum
spanning tree.  We will later also examine the
$\mathbf{NP}$-\textsc{Complete} Steiner network problem.

Random graphs are generated using the Erd\H{o}s-R\'enyi model $G(n,m)$
with $n \in \{4, \ldots, 30\}$ vertices and $m \in \{3, \ldots,
{n \choose 2}\}$ edges.  The edges of the random graphs are weighted
according to various probability distributions (listed in the
forthcoming figures, below).  A random feasible solution is chosen by
solving the minimum spanning tree problem on the graph
while \emph{ignoring} edge weights.  While this will in effect
produce a ``random'' feasible solution that is independent of the edge
weights (and therefore also the optimization problem's objective
function), note that this method is not equivalent to drawing
uniformly from the set of all possible feasible solutions.  As such,
these results may slightly bias the random solutions' optimality.
The optimal solution is discovered by solving the problem taking edge
weights into account.  The constant of approximation is then
calculated by taking the ratio of the cost of the random feasible
solution over the cost of the optimal solution.  As we can see in
Figure~\ref{fig:randomst}, the randomly chosen solutions rapidly
converge to an approximation constant of 2.

\begin{figure}
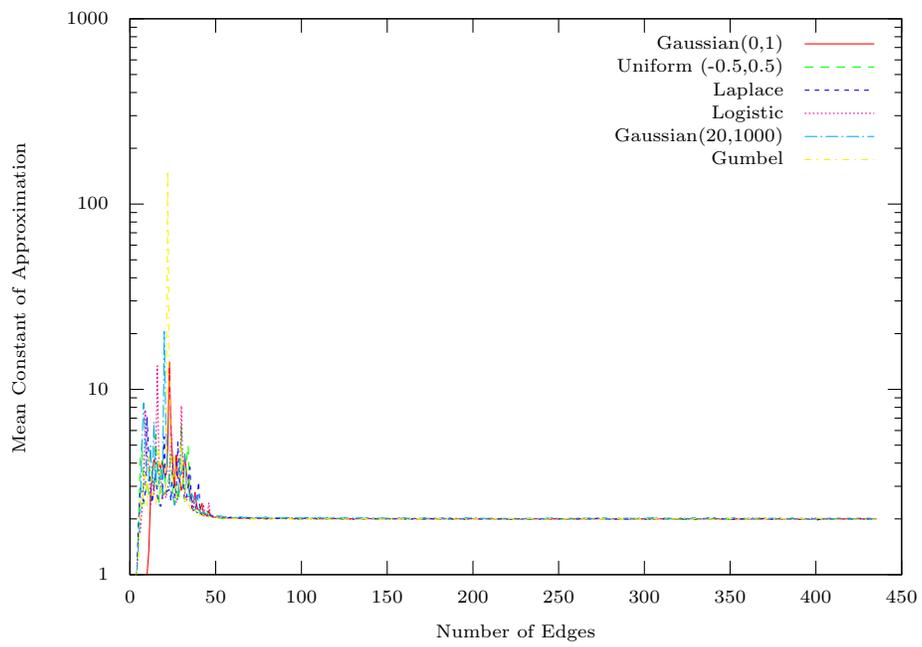
\scriptsize


  \caption{The average constant of approximation of randomly chosen feasible solutions to the minimum spanning tree problem converges to 2.}
  \label{fig:randomst}
\end{figure}

Next, let us consider the Steiner network problem.  The random graphs
are generated similarly to the minimum spanning tree case above.  The
number of terminals in the problem is set to $\frac{n}{2}$ and the
terminals are chosen randomly from the vertices according to a uniform
distribution.  The random feasible solutions are chosen by solving the
minimum spanning tree problem on the graph.  Since the Steiner network
problem is computationally intractable to solve optimally, though, we
are forced to compare the randomly chosen feasible solutions against a
bounded approximation to the optimal solution using a 2-Optimal
algorithm~\cite{goemans94improved}.  Therefore, if the random feasible
solutions have a cost no worse than the approximation algorithms, we
know that the cost of the feasible solutions are no more than twice
the cost of the optimal solution.  As we see in
Figure~\ref{fig:randomsteiner}, for almost all graphs the randomly
chosen solutions actually perform better than the approximation
algorithm, with improvement as the graphs get denser.

\begin{figure}
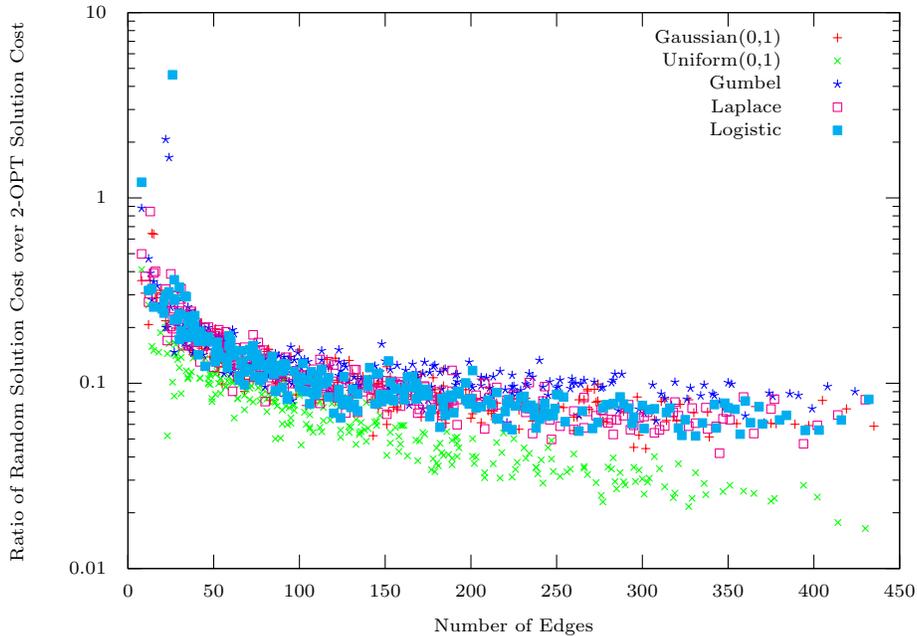
\scriptsize


  \caption{Random feasible solutions to the Steiner network problem compared to a 2-OPT algorithm.  Values less than 1 imply that the random solution was lower cost than the approximation algorithm's solution.}
  \label{fig:randomsteiner}
\end{figure}

\section{Conclusions}

This paper investigated and bounded the expected solution quality of
combinatorial optimization problems when feasible solutions are chosen
at random.  A loose general bound was discovered
(Theorem~\ref{thm:finalbounds}), as well as families of combinatorial
optimization problems for which random feasible solutions are expected
to be a constant factor of optimal.  One implication of this result is
that, for graphical problems, if the average edge weight in a feasible
solution is sufficiently small, then any randomly chosen feasible
solution to the problem will be a constant factor of optimal.  For
example, under certain well-defined circumstances, the expected
constant of approximation of a randomly chosen feasible solution to
the Steiner network problem is bounded above by 3.  Empirical analysis
supports these bounds and actually suggest that they might be
tightened.

\section*{Acknowledgements}

This research was performed under The Johns Hopkins University Applied
Physics Laboratory research grant X9RSTHNA and the Stuart S. Janney
award.  Thanks are owed to Kevin Schultz and Ali Shokoufandeh for
their helpful comments and insight.

\appendix


\section{Proofs and Supplementary Theorems}
\label{appendix:proofs}

Let $H_{n,r}$ denote the generalized harmonic number of order $n$ of $r$:
\begin{displaymath}
  H_{n,r} = \sum_{i=1}^n \frac{1}{i^r}.
\end{displaymath}
\begin{lemma}
  \label{lem:bounds}
  \begin{displaymath}
    2\sqrt{n+1}-2 < H_{n,\frac{1}{2}} \leq 2\sqrt{n}-1.
  \end{displaymath}
\end{lemma}
\begin{proof}
  Since $\frac{1}{\sqrt{x}}$ is monotonically decreasing in $x$, we can bound $H_{n,\frac{1}{2}}$ by its definite integral using the forward and backward rectangle rules.
  For the lower bound,
  \begin{eqnarray*}
    H_{n,\frac{1}{2}} & = & \sum_{i=1}^n \frac{1}{\sqrt{i}}\\
    & > & \int_1^{n+1} \frac{1}{\sqrt{x}}\,dx\\
    & = & 2\sqrt{n+1} - 2.
  \end{eqnarray*}
  For the upper bound,
  \begin{eqnarray*}
    H_{n,\frac{1}{2}} & = & 1 + \sum_{i=2}^n \frac{1}{\sqrt{i}}\\
    & \leq & 1 + \int_1^n \frac{1}{\sqrt{x}}\,dx\\
    & = & 2\sqrt{n} - 1.
  \end{eqnarray*}  
\end{proof}
\begin{proof}[Proof of Theorem~\ref{thm:mlb}]
  The conditions of the implication in the theorem ensure that the bounds of Theorem~\ref{thm:origbounds} simplify to
  \begin{displaymath}
    E[X_{(r:\numvars)}] \geq \mu - \sigma\sqrt{\frac{\numvars}{2r}}.
  \end{displaymath}
  Applying this bound to the trimmed sum of the $m$ smallest order statistics results in
  \begin{displaymath}
    m\mu - \sum_{r=1}^m \sigma\sqrt{\frac{\numvars}{2r}}.
  \end{displaymath}
  Further simplification yields
  \begin{equation}
    \label{equ:msimplified}
    m\mu - \frac{1}{2}\sigma\sqrt{2\numvars}\sum_{r=1}^m \sqrt{\frac{1}{r}}.
  \end{equation}
  By Lemma~\ref{lem:bounds}, (\ref{equ:msimplified}) is bounded below by
  \begin{displaymath}
    m\mu - \sigma\sqrt{2\numvars}\left(\sqrt{m+1} - 1\right).
  \end{displaymath}
\end{proof}
\begin{proof}[Proof of Theorem~\ref{thm:lub}]
  The conditions of the implication in
  the theorem ensure that the bounds of Theorem~\ref{thm:origbounds}
  simplify to
  \begin{displaymath}
    E[X_{(r:\numvars)}] \leq \mu + \sigma\sqrt{\frac{\numvars}{2(\numvars-r+1)}}.
  \end{displaymath}
  Applying this bound to the trimmed sum of the $\ell$ largest order statistics results in
  \begin{displaymath}
    \ell\mu + \sum_{r=\numvars-\ell+1}^\numvars \sigma\sqrt{\frac{\numvars}{2\numvars - 2r + 2}}.
  \end{displaymath}
  Further simplification yields
  \begin{equation}
    \label{equ:lsimplified}
    \ell\mu + \frac{1}{2}\sigma\sqrt{2\numvars}\sum_{r=\numvars-\ell+1}^\numvars \sqrt{\frac{1}{\numvars - r + 1}}.
  \end{equation}
  The summation in (\ref{equ:lsimplified}) is equal to $H_{\ell,\frac{1}{2}}$:
  \begin{eqnarray*}
    \sum_{r=\numvars-\ell+1}^\numvars \sqrt{\frac{1}{\numvars - r + 1}} & = &  \sqrt{\frac{1}{\ell}} + \sqrt{\frac{1}{\ell - 1}} + \sqrt{\frac{1}{\ell - 2}} + \ldots + \sqrt{\frac{1}{1}}\\
    & = & \sum_{i=1}^{\ell}\frac{1}{\sqrt{i}}\\
    & = & H_{\ell,\frac{1}{2}},
  \end{eqnarray*}
  which, by Lemma~\ref{lem:bounds}, is bounded above by $2\sqrt{\ell}-1$.
\end{proof}
\begin{proof}[Proof of Theorem~\ref{thm:finalbounds}]
  Let $\overline{E[Y]}$ be the upper bound of Theorem~\ref{thm:lub} and $\underline{E[Y^*]}$ be the lower bound of Theorem~\ref{thm:mlb}.  Note that
  \begin{displaymath}
    1 + \left|\frac{E[Y] - E[Y^*]}{E[Y^*]}\right| \leq 1 + \left|\frac{\overline{E[Y]} - \underline{E[Y^*]}}{\underline{E[Y^*]}}\right|,
  \end{displaymath}
  allowing us to proceed by bounding the expression using $\overline{E[Y]}$ and $\underline{E[Y^*]}$.
  Next, observe that 
  $\overline{E[Y]} \geq \underline{E[Y^*]}$, so the absolute value in the numerator can be dropped:
  \begin{displaymath}
    1 + \left|\frac{E[Y] - E[Y^*]}{E[Y^*]}\right| \leq 1 + \frac{\overline{E[Y]} - \underline{E[Y^*]}}{\left|\underline{E[Y^*]}\right|}.
  \end{displaymath}
  There are therefore two cases to consider, depending on the sign of the denominator.
  Let us first consider the case when $\underline{E[Y^*]} < 0$.  Since $m$, $k$, and $\sigma$ are necessarily positive, observe that
  \begin{displaymath}
    \underline{E[Y^*]} < 0 \iff \mu < \sigma\sqrt{2\numvars}\frac{\sqrt{m+1}-1}{m},
  \end{displaymath}
  producing the first case in the piecewise function of the theorem.
  Under these circumstances, evaluating Equation~\ref{equ:approxbound} yields
  \begin{multline}
    \label{equ:firstcase}
    1 + \left|\frac{E[Y] - E[Y^*]}{E[Y^*]}\right| \leq 1 + \frac{\overline{E[Y]} - \underline{E[Y^*]}}{\underline{E[Y^*]}}\\= \displaystyle\frac{\ell\mu - 2m\mu + \sigma\sqrt{2\numvars}\left(\sqrt{\ell}+2\sqrt{m+1}-\frac{5}{2}\right)}{\sigma\sqrt{2\numvars}\left(\sqrt{m+1}-1\right)m\mu}.
  \end{multline}
  Since $m \geq 1 \implies \sigma\sqrt{2\numvars}\left(\sqrt{m+1}-1\right)\geq \sigma\numvars$ and
  \begin{displaymath}
    \ell, m \geq 1 \implies \sigma\sqrt{2\numvars}\left(\sqrt{\ell}+2\sqrt{m+1}-\frac{5}{2}\right) \leq 3\numvars,
  \end{displaymath}
  the bound can be relaxed to
  \begin{displaymath}
    (\ref{equ:firstcase}) \leq \frac{\ell\mu - 2m\mu + 3\sigma\numvars}{\sigma\numvars-m\mu}.
  \end{displaymath}

  In the second and final case, $\underline{E[Y^*]} > 0$ and 
  \begin{multline*}
     1 + \left|\frac{E[Y] - E[Y^*]}{E[Y^*]}\right| \leq 1 + \frac{\overline{E[Y]} - \underline{E[Y^*]}}{\underline{E[Y^*]}}\\= \displaystyle\frac{\ell\mu + \sigma\sqrt{2k}(\sqrt{\ell} - \frac{1}{2})}{m\mu - \sigma\sqrt{2k}(\sqrt{m+1}-1)}  \leq \frac{\ell\mu + \sigma k}{m\mu - \sigma k},
  \end{multline*}
  because $2(\ell-1), 2m \leq k$ implies $\sqrt{2k}(\sqrt{\ell}+2\sqrt{m+1} - \frac{5}{2}) \leq k$ and $\sqrt{2k}(\sqrt{m+1}-1) \leq k$.
\end{proof}

\begin{proof}[Proof of Corollary~\ref{cor:ub}]
 Evaluating the expression yields
  \begin{eqnarray*}
    \frac{2\ell\mu + \sigma\sqrt{2\numvars}\left(1 + \log\ell\right)}{2m\mu + 2\sigma\sqrt{2\numvars}\left(\sqrt{m+1}-1\right)} & \leq & \frac{\ell}{m}\\
    & \Downarrow &\\
    2\mu + \frac{\sigma\sqrt{2\numvars}\left(\log\ell + 1\right)}{\ell} & \leq & 2\mu + \frac{2\sigma\sqrt{2\numvars}\left(\sqrt{m+1}-1\right)}{m}\\
    & \Downarrow &\\
    \frac{\log\ell+1}{\ell} & \leq & \frac{2\left(\sqrt{m+1}-1\right)}{m}\\
    & \Downarrow &\\
    \ell & \leq & \exp(2\sqrt{m+1}-3),
  \end{eqnarray*}
  which is satisfied by the conditions of the corollary.
\end{proof}

\begin{proof}[Proof of Corollary~\ref{cor:asymp}]
  From Theorem~\ref{thm:finalbounds},
  \begin{displaymath}
    \sigma \gg \mu \implies 1 + \left|\frac{E[Y] - E[Y^*]}{E[Y^*]}\right| \leq \frac{\ell\mu - 2m\mu + 3\sigma k}{\sigma k - m\mu} = \Theta(3).
  \end{displaymath}
\end{proof}

\bibliographystyle{elsarticle-num}
\bibliography{random}

\end{document}